\documentclass[conference,a4paper]{IEEEtran}
\usepackage[utf8]{inputenc}
\usepackage[T1]{fontenc}

\usepackage{cite}
\usepackage{array}
\usepackage{amssymb,amsmath,amsthm,enumerate}

\def\C{\mathbb{C}} 
\def\R{\mathbb{R}} 
\def\wp1{\mathrm{w.p.} 1}  
\def\Prob{\mathbb{P}}
\def\Exp{\mathbb{E}} 
\def\kS{$k$-space}

\newtheorem{thm}{Theorem}
\newtheorem{prop}{Proposition}
\newtheorem{lemma}{Lemma}
\newtheorem{remark}{\it {Remark}}

\usepackage{xcolor}

\usepackage{graphicx}

\def\cc#1{\setlength{\tabcolsep}{0pt}\begin{tabular}{c}#1\end{tabular}}

\def\yaxis#1{\cc{\rotatebox{90}{{\small #1}}}}

\begin{document}

\title{From variable density sampling to continuous sampling using Markov chains}

\author{\IEEEauthorblockN{Nicolas Chauffert, Philippe Ciuciu}
\IEEEauthorblockA{CEA, NeuroSpin center, \\
INRIA Saclay, PARIETAL Team \\
145, F-91191 Gif-sur-Yvette, France\\
Email: firstname.lastname@inria.fr}
\and
\IEEEauthorblockN{Pierre Weiss}
\IEEEauthorblockA{ITAV, USR 3505\\
Toulouse, France\\
Email: pierre.weiss@itav-recherche.fr}
\and
\IEEEauthorblockN{Fabrice Gamboa}
\IEEEauthorblockA{Universit\'e de Toulouse; CNRS\\ IMT-UMR5219\\F-31062 Toulouse, France\\
Email: fabrice.gamboa@math.univ-toulouse.fr}}

\maketitle

\begin{abstract}
Since its discovery over the last decade, Compressed Sensing (CS) has been successfully applied
to Magnetic Resonance Imaging (MRI). It has been shown to be a powerful way to reduce scanning time without sacrificing image quality. MR images are actually strongly compressible in a wavelet basis, the latter being largely incoherent with the $k$-space or spatial Fourier domain where acquisition is performed. Nevertheless, since its first application to MRI~\cite{Lustig07}, the theoretical justification of actual $k$-space sampling strategies is questionable. Indeed, the vast majority of $k$-space sampling distributions have been heuristically designed~(e.g., variable density) or driven by experimental feasibility considerations~(e.g., random radial or spiral sampling to achieve smoothness $k$-space trajectory). In this paper, we try to reconcile very recent CS results with the MRI specificities~(magnetic field gradients) by enforcing the measurements, i.e. samples of $k$-space, to fit continuous trajectories. 
To this end, we propose random walk continuous sampling based on Markov chains and we compare the reconstruction quality of this scheme to the state-of-the art.
\end{abstract}


\section{Introduction}

Compressed Sensing \cite{Candes06,Donoho06} is a theoretical framework which gives guarantees to recover sparse signals (signals reprensented by few non-zero coefficients in a given basis) from a limited number of linear projections.
In some applications, the measurement basis is fixed and the projections should be selected amongst a fixed set. 
For instance, in MRI, the signal is sparse in the wavelet basis, and the sampling is performed in the spatial~(2D or 3D) Fourier basis (called \kS). Possible measurements are then projections on the lines of matrix $A=F^* \Psi$, where $F^*$ and $\Psi$ denote the Fourier and inverse wavelet transform, respectively. 

Recent results~\cite{Rauhut10,Candes11} give bounds on the number of measurement $m$ needed to exactly recover $s$-sparse signals in $\C^n$ or $\R^n$ in the framework of bounded orthogonal systems. The authors have shown that for a given $s$-sparse signal, the number of measurements needed to ensure its perfect recovery is $O(s \log(n))$. This methodology, called \textit{variable density sampling}, involves an independent and identically distributed~(iid) random drawing and has already given promising results in reconstruction simulations~\cite{Lustig07,Puy11}.  Nevertheless, in real MRI, such sampling patterns cannot be implemented, because of the limited speed of magnetic field gradient commutation. Hardware constraints require at least continuity of the sampling trajectory, which is not satisfied by two-dimensional iid sampling. In this paper, we introduce a new Markovian sampling scheme to enforce continuity. Our approach relies on the following reconstruction condition introduced by Juditski, Karzan and Nemirovki~\cite{Juditsky11}:
\begin{thm}[\cite{Juditsky11}]
If $A$ satisfies:
\begin{equation*}
\gamma(A)=\min_{Y\in \R^{n\times m}} \|I_n-Y^T A\|_\infty < \frac{1}{2 s}.
\end{equation*}  
All $s$-sparse signals $x \in \R^n$ are recovered exactly by solving:
\begin{equation}
\label{eq:minL1}
\underset{A_m w=A_m x}{\operatorname{argmin}} \ \|w\|_1
\end{equation}
\end{thm}
\noindent which can be seen as an alternative to the \textit{mutual coherence}~\cite{Donoho06}. 
We will show that this criterion makes it possible to obtain theoritical guarantees on the number of measurements necessary to reconstruct $s$ sparse signals, using variable density sampling or markovian sampling. Unfortunately the bounds we obtain are in $O(s^2)$. This phenomenon is due to the \textit{quadratic bottleneck} described in \cite{Rauhut10}. We are currently trying to obtain $O(s)$ results using different proof strategies.

\subsection*{Notation}

A signal $x \in \R^n$ is said to be $s$-sparse if it has at most $s$ non-zero coefficients. $x$ is measured through the acquisition system represented by a matrix $A_0$. Downsampling the measurements consists of deriving a matrix $A$ composed of $m$ lines of $A_0$ and observing $y=Ax \in \R^m$. 

\section{\label{part:2}Theoretical result}

\subsection{Independent Sampling}
We aim at finding $A_m \in \R^{m \times n}$ composed of $m$ rows of $A$, and $Y_m \in \R^{m\times n}$ such that $\|I_n-Y_m^T A_m \|_\infty < \frac{1}{2 s}$, for a given positive integer $s$. 
Following \cite{Juditsky11b}, we set $\Theta_i=\frac{a_i a_i^T}{\pi_i}$ and use the decomposition $I_n = A^TA = \sum_{i=1}^{n} \pi_i \Theta_i$.
We consider a sequence of $m$ random i.i.d. matrices $Z_1 , \dots, Z_m$, taking value $\Theta_i$ with probability $\pi_i$.
We set $\pi_i = \|a_i\|_\infty^2/L$, where $L = \sum_{i=1}^n \|a_i\|_\infty^2$, so that $\|Z_l\|_\infty$ is equal to $L$. 
Let us denote $W_m = \frac{1}{m} \sum_{l=1}^m Z_l$. Then $W_m$ may be written as $Y_m^T A_m$.
\begin{lemma}
\label{prop:indep}
$\forall t >0$ \\
\begin{equation}
\label{eq:conc_indep}
\Prob(\|I_n - W_m\|_\infty >t) \leq n (n+1) \exp \Bigl(- \frac{m t^2}{2 L^2 + 2 Lt/3} \Bigr).
\end{equation}
\end{lemma}

\begin{proof}
Bernstein's concentration inequality \cite{ledoux01} states that if $X_1, \dots, X_m$ are independent zero-mean random variables such that for all $i$, $|X_i| \leq \alpha$ and $\displaystyle \sigma^2= \sum_i \Exp  \left(X_i^2\right) < \infty$, then  $\forall t>0$
\begin{equation*}
\Prob \left( |\sum_{i=1}^{m} X_i| >t \right) \leq 2 \exp \left(- \frac{t^2}{2(\sigma^2+\alpha t/3)} \right).
\end{equation*}  
For $1\leq a, b \leq n$, let $M^{(a,b)}$ denote the $(a,b)$th entry of a matrix $M \in \R^{n \times n}$. The random variable $(I_n-Z_l)^{(a,b)}$ is centered since $\sum_{i=1}^n \pi_i \Theta_i = I_n$.
Moreover, $|(I_n-Z_l)^{(a,b)}| \leq L$. Applying Bernstein's inequality to the sequence $\frac{1}{m}\left((I_n-Z_l)^{(a,b)}\right)_{1\leq l\leq m}$ gives 
\begin{equation*}
\Prob \left( |(I_n-W_m)^{(a,b))} | >t\right) \leq 2 \exp \left(- \frac{m t^2}{2 L^2 + 2Lt/3 } \right).
\end{equation*}
Finally, using a union bound and the symmetry property of matrix $(I_n-W_m)$, we get:
\begin{equation}
\label{eq:ineg_inf}
\Prob \left(\|I_n-W_m\|_\infty > t\right) \leq \sum_{1\leq a \leq b \leq n} \Prob \left(|I_n-W_m|^{(a,b)} > t\right).
\end{equation}
Since $\Prob \left(|I_n-W_m|^{(a,b)} > t\right)$ is independent of $(a,b)$, we obtain Eq.~\eqref{eq:conc_indep}.
\end{proof}

\begin{remark}
Setting $t=4 L \sqrt{\frac{2\ln (2 n^2)}{m}}$ in lemma~\ref{prop:indep}, the bound given by Juditsky et al. in \cite{Juditsky11b} is $\Prob \left( \|I_n-W_m\|_\infty \geq t \right) \leq \frac{1}{2}$. This bound is obtained by upper-bounding the mean of $\|I_n-W_m\|_\infty$ and using Markov inequality.
Setting the same $t$ value in Eq.~\eqref{eq:conc_indep}, and assuming $t\leq L$, we obtain $\Prob \left( \|I_n-W_m\|_\infty \geq t \right) \leq \frac{1}{2 n^{4}}$.
This huge difference comes from inability of Markov inequality to capture large deviations behaviors. 
\end{remark}
From lemma~\ref{prop:indep}, we can derive the immediate following result by setting $t=1/2s$:

\begin{prop}
Let $A_m$ be a measurement matrix designed by drawing $m$ lines of $A$ under the distribution $\pi$. Then, with probability $1-\eta$, if
\begin{equation}
m \geqslant 5 L^2 s^2 \log(n^2/\eta),
\end{equation}
every $s$-sparse signal $x$ is the unique solution of the $\ell_1$ problem:
\begin{equation*}
\underset{A_m w=A_m x}{\operatorname{argmin}} \ \|w\|_1
\end{equation*} 
\end{prop}

\subsection{Markovian sampling}

Sampling patterns obtained using the strategy presented in Section \ref{part:2} are not usable for many practical devices. 
A common constraint met on many hardwares (e.g. MRI) is the proximity of successive measurements. 
A simple way to model dependence between successive samples consists of introducing a Markov chain $X_1 \dots X_m$ on the set $\{1, \dots, n\}$ that represents locations of possible measurements. The transition probability to go from location $i$ to location $j$ is positive if and only if sampling $i$ and $j$ successively is possible. We denote $W_m=\frac{1}{m}\sum_{l=1}^{m} \Theta_{X_l}$.\\
In order to use a concentration inequality, $W_m$ should satisfy $\Exp  \left(W_m\right)=I_n$. We thus need (i) to set the stationary distribution of the Markov chain to $\pi$ and (ii) to set up the chain with its stationnary distribution $\pi$. These two conditions ensure that the marginal distribution of the chain is $\pi_i$ at any time. The issue of designing such a chain is widely studied in the frame of Markov chain Monte Carlo (MCMC) algorithms. 

A simple way to build up the transition matrix $P= (P_{ij})_{1\leq i,j \leq n}$ is the Metropolis
algorithm~\cite{hastings1970montecarlo}. Let us now recall a concentration inequality for finite-state Markov chains~\cite{Lezaud98}.
\begin{thm}
\label{thm:Lezaud}
Let $(P,\pi)$ be an irreductible and reversible Markov chain on a finite set G of size $n$. Let $f:G \rightarrow \mathbb{R}$ be such that $\sum_{i=1}^n\pi_i f_i = 0 , \, \|f\|_\infty \leq 1$ and $0< \sum_{i=1}^n f_i^2 \pi_i \leq b^2$. Then, for any initial distribution $q$, any positive integer $m$ and all $0< t\leq 1$,
\begin{equation*}
\Prob \Bigl(\frac{1}{m} \sum_{i=1}^m f(X_i) \geq t \Bigr) \leq e^{\frac{\epsilon(P)}{5}} N_q \exp \Bigl(- \frac{m t^2 \epsilon(P)}{4 b^2(1+h(5 t/b^2))} \Bigr)
\end{equation*}
where $N_q=(\sum_{i=1}^n (\frac{q_i}{\pi_i})^2 \pi_i)^{1/2}$, $\beta_1(P)$ is the second largest eigenvalue of $P$, and $\epsilon(P)=1-\beta_1(P)$ is the spectral gap of the chain. Finally $h$ is given by $h(x)=\frac{1}{2}(\sqrt{1+x} - (1-x/2))$.
\end{thm}
Using this theorem, we can guarantee the following control of the term $\|I_n - W_m\|_\infty$:\\
\begin{lemma}
$ \forall\  0<t\leq1$,
\begin{equation}
\label{eq:conc_markov}
 \Prob \left(\|I_n - W_m\|_\infty \! \geq t \right) \! \leq \! n (n+1) e^{\frac{\epsilon(P)}{5}} \! \exp \Bigl(\!- \frac{mt^2 \epsilon(P)}{12L^2}\Bigr).
\end{equation}
\end{lemma}

\begin{proof}
By applying Theorem~\ref{thm:Lezaud} to a function $f$ and then to its opposite $-f$, we get:

\begin{multline*}
\Prob \Bigl(\Bigl|\frac{1}{m} \sum_{i=1}^m f(X_i)\Bigr| \geq t \Bigr) \leq 2 e^{\frac{\epsilon(P)}{5}} N_q \\ \exp \Bigl(- \frac{m t^2 \epsilon(P)}{4 b^2(1+h(5 t/b^2))} \Bigr).
\end{multline*}

Then we set $f(X_i)=(I_n-\Theta_{X_i})^{(a,b)}/(1+L)$. 
The Markov chain is constructed such that $\sum_{i=1}^n\pi_i f(X_i)=0$. 
Since we have $\|f\|_\infty \leq 1$, $b=1$, and since $t\leqslant 1$, $1+h(5t)<3$.
Moreover, since the initial distribution is $\pi$, $q_i=\pi_i, \forall i$ and thus $N_q=1$. Again, resorting to a union bound (\ref{eq:ineg_inf}) enables us to extend the result for the $(a,b)$th entry to the whole infinite norm of the $n \times n$ matrix $I_n-W_m$~\eqref{eq:conc_markov}.\\ \end{proof}

Then we can quantify the number of measurements needed to ensure exact recovery:
\begin{prop}
\label{prop:measurements_needed}
Let $A_m$ be a measurement matrix designed by drawing $m$ lines of $A$ under the Markovian process described above. Then, with probability $1-\eta$, if
\begin{equation}
\label{eq:measurements}
m \geqslant \frac{12 L^2}{\epsilon(P)} s^2 \log(2n^2/\eta),
\end{equation}
every  $s$-sparse signal $x$ is the unique solution of the $\ell_1$ problem:
\begin{equation*}
\underset{A_m w=A_m x}{\operatorname{argmin}} \ \|w\|_1
\end{equation*} 
\end{prop}

\begin{remark}
The spectral gap $\epsilon(P)$ takes its value between 0 and 1 and describes the mixing properties of the Markov chain. 
The closer the spectral gap to 1, the fastest the convergence to the mean. \\
\end{remark}

\begin{remark}
All the results above can be extended to the complex case using a slightly different proof.
\end{remark}

\section{Results and discussion}

In order to cover a larger domain of $k$-space, we consider the following chain: $P^{(\alpha)}=(1-\alpha)P+ \alpha\tilde{P}$, where $\tilde{P}$ corresponds to an independent drawing $\tilde{P}_{ij}=\pi_j,\forall i,j$. This chain has $\pi$ as invariant distribution, and fulfills the continuity property while enabling a jump with probability of $\alpha$.

Weyl's Theorem~\cite{Horn91} ensures that $\epsilon(P^{(\alpha)}) > \alpha$. This bound is useful because of the dependence of $\epsilon(P)$ with respect to the problem dimension, which would have weakened condition~\eqref{eq:measurements}.

Sampling scheme obtained by these methods are composed of $1/\alpha$-average length random walks on the $k$-space. All our experiments consist of reconstructing a two-dimensional image from a sampled $k$-space by solving an $\ell_1$ minimization problem. Constrained $\ell_1$ minimization (Eq.~\eqref{eq:minL1}) is performed using the Douglas-Rachford algorithm~\cite{Combettes11b}. In each case, only twenty percent of the Fourier coefficients are kept, which corresponds to an acceleration factor of $r=5$. Since the schemes are obtained by a random process, we run each experiment 10 times independently, and compared the mean value of the reconstruction results in terms of \textit{Peak Signal-to-Noise Ratio} (PSNR).\\

In Fig.~\ref{fig:1}, it is shown that the image reconstruction quality degrades when $\alpha$ decreases. These results can be explained by the spatial confinement of the continuous parts of a given Markov chain, except for large values of $\alpha$. There seems to be a compromise between the number of discontinuities of the chain (linked to the hardware constraints in MRI) and the $k$-space coverage. Nevertheless, accurate reconstruction results can be observed with reasonable average mean length of connected subparts ($\alpha=0.01$ or $0.001$).  

The mixing properties of the chain (through its spectral gap) seem to have a strong impact on the quality of the scheme, as shown in Proposition~\ref{prop:measurements_needed}. Unfortunately, the spectral gap is strongly related to the problem dimension $n$ and can tend to zero if $n$ goes to infinity. This proves to be a theoretic limitation of this method. Nevertheless, we obtained reliable reconstruction results which cannot be explained by the proposed theory. Since the design process is based on randomness, we can even expose a specific scheme which provides accurate reconstruction results instead of considering the mean behavior (Fig.~\ref{fig:2}). We currently aim at deriving a stronger result on the number of measurements needed, involving a $O(s)$ bound. Meanwhile, we are developing second order chains which can ensure more regularity of the trajectories and for which we have already observed good reconstruction results (Fig.~\ref{fig:3}).

\begin{figure}
\begin{center}
\begin{tabular}{cc}
\yaxis{$k_y$} \includegraphics[width=.2\textwidth]{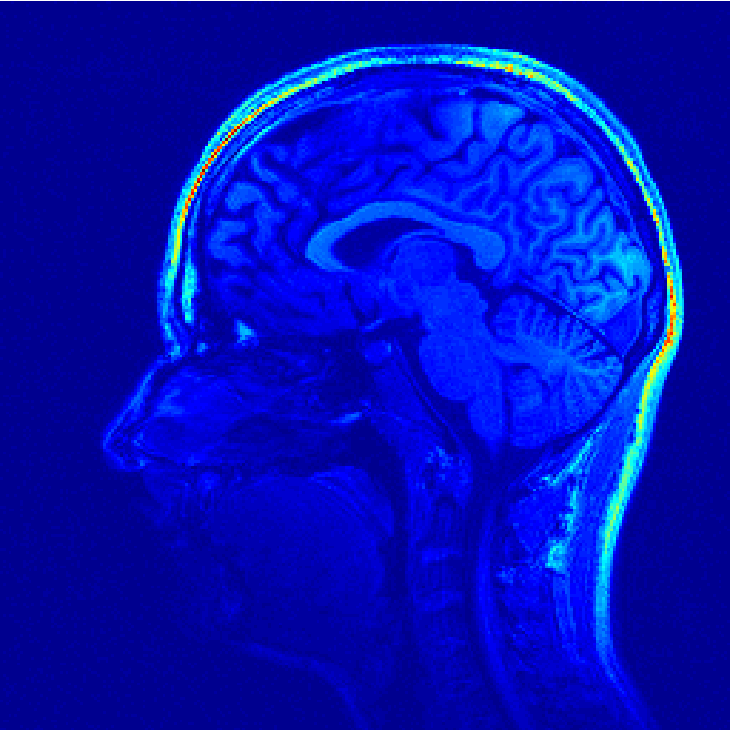} & \includegraphics[width=.2\textwidth]{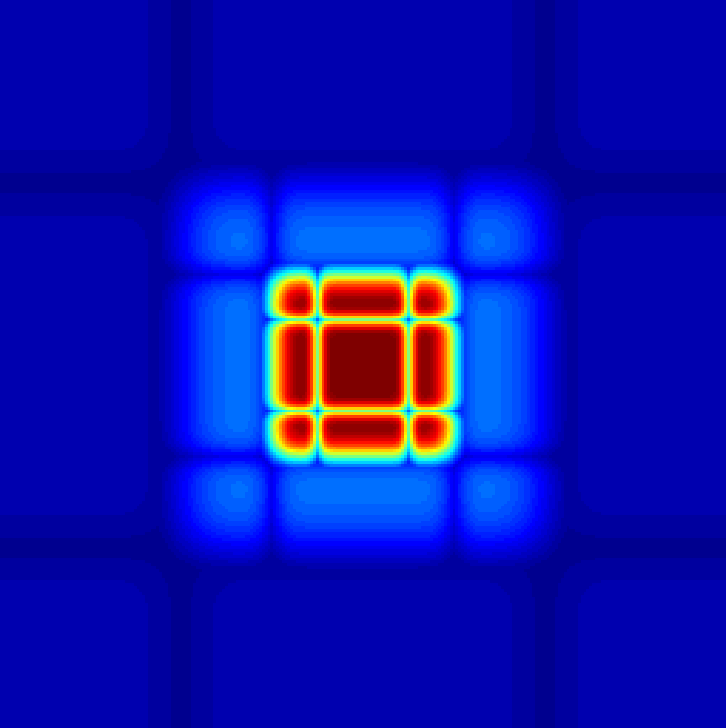} \includegraphics[height=.2\textwidth]{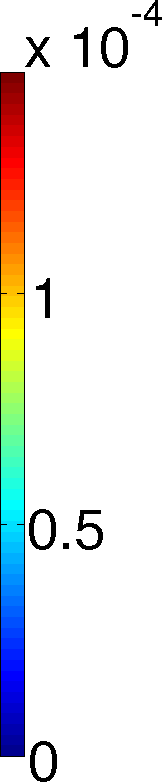}\\ [-4.2cm]
{\small (a)} & {\small (b)}\\[4.3cm]
\yaxis{$k_y$} \includegraphics[width=.2\textwidth]{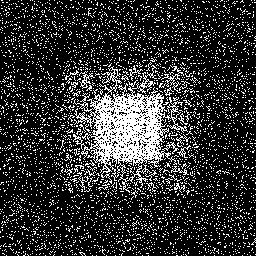} & \includegraphics[width=.2\textwidth]{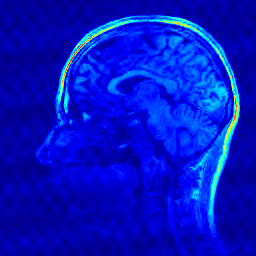}\\ [-4.2cm]
{\small (c) $\alpha=1$} & {\small (d) mean-PSNR=33.4dB}\\[4.3cm]
\yaxis{$k_y$} \includegraphics[width=.2\textwidth]{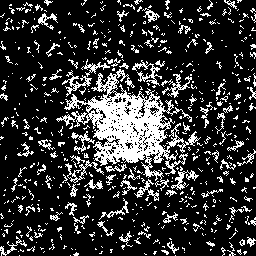} & \includegraphics[width=.2\textwidth]{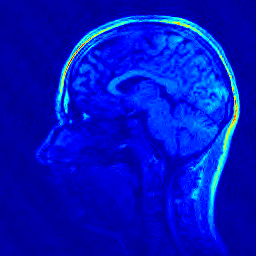}\\ [-4.2cm]
{\small (e) $\alpha=0.1$} & {\small (f) mean-PSNR=32.4dB}\\[4.3cm]
\yaxis{$k_y$} \includegraphics[width=.2\textwidth]{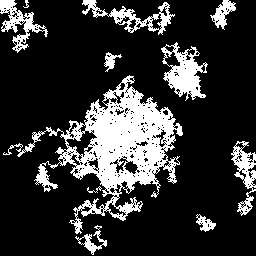} & \includegraphics[width=.2\textwidth]{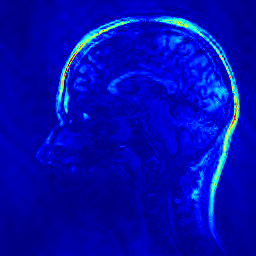}\\ 
 $k_x$ & \\[-4.6cm]
{\small (g) $\alpha=0.001$} & {\small (h) mean-PSNR=30.3dB}\\[4.3cm]
\end{tabular}\vspace*{-.5cm}
\end{center}
\caption{{\bf First line:} reference image used in our experiments (a) and $\pi$ distribution (b). {\bf Lines 2 to 4, left:} different sampling patterns (with an acceleration factor $r=5$). {\bf right:} reconstruction results. From line 2 to bottom: independent drawing from distribution $\pi$~(c), corresponding to $\alpha=1$. (e) (resp (g)) represents a sampling scheme designed with the presented markovian process with transition matrix $P^{(\alpha)}$ for $\alpha=0.1$) (resp. $\alpha=0.001$).\label{fig:1}}
\end{figure}

\begin{figure}[!h]
\begin{center}
\begin{tabular}{cc}
\yaxis{$k_y$} \includegraphics[width=.2\textwidth]{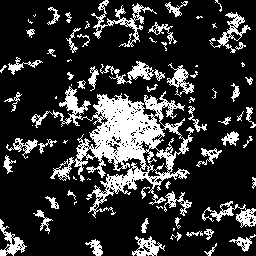} & \includegraphics[width=.2\textwidth]{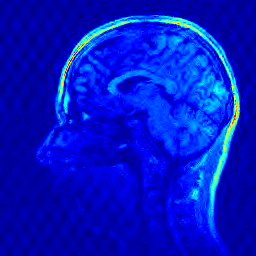}\\ 
 $k_x$ & \\[-4.6cm]
{\small (a) $\alpha=0.01$} & {\small (b) PSNR=34.2dB}\\[4.3cm]
\end{tabular}\vspace*{-.5cm}
\end{center}
\caption{Sampling scheme obtained setting $\alpha=0.01$ and $r=5$ (a) and its corresponding reconstructed image (b). \label{fig:2}}
\end{figure}

\begin{figure}[!h]
\vspace*{.3cm}
\begin{center}
\begin{tabular}{cc}
\yaxis{$k_y$} \includegraphics[width=.2\textwidth]{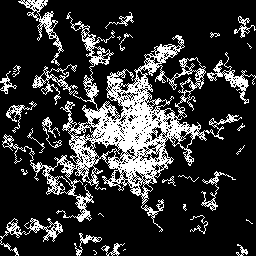} & \includegraphics[width=.2\textwidth]{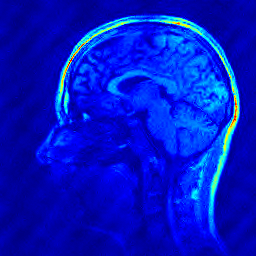}\\ 
 $k_x$ & \\[-4.6cm]
{\small (a) $\alpha=0.01$} & {\small (b) PSNR=33.4dB}\\[4.3cm]
\end{tabular}\vspace*{-.5cm}
\end{center}
\caption{Preliminary results for second order Markov chain: sampling scheme obtained setting $\alpha=0.01$ and $r=5$ (a) and its corresponding reconstructed image (b). \label{fig:3}}
\end{figure}

\section{Conclusion}
We proposed a novel approach combining compressed sensing and Markov chains to design continuous sampling trajectories, required for MRI applications. Our work may easily be extended to a 3D framework by considering a different neighbourhood of each $k$-space location. Existing continuous trajectories in CS-MRI only exploit 1D or 2D randomness for 2D or 3D $k$-space sampling, respectively. In the latter case, the points are randomly drawn in the plane defined by the partition and phase encoding directions so as to maintain continuous sampling in the orthogonal readout direction (frequency encoding). Here, the novelty relies both \newpage \noindent on the use of randomness in all $k$-space dimensions, and the establishment of compressed sensing results for continuous trajectories, based on a concentration result for Markov chains.

\section*{Acknowledgements}
We thank J\'er\'emie Bigot for the time he dedicated to our questions and his helpful remarks. The authors would like to thank the CIMI Excellence Laboratory for inviting Philippe Ciuciu on an excellence researcher position during winter 2013.

{\footnotesize
\bibliographystyle{IEEEbib}

\begin{thebibliography}{10}

\bibitem{Lustig07}
M.~Lustig, D.~L. Donoho, and J.~M. Pauly,
\newblock ``Sparse {MRI}: The application of compressed sensing for rapid {MR}
  imaging,''
\newblock {\em {{M}agn. {R}eson. {M}ed.}}, vol. 58, no. 6, pp. 1182--1195, Dec.
  2007.

\bibitem{Candes06}
E.~Cand{\`e}s, J.~Romberg, and T.~Tao,
\newblock ``Robust uncertainty principles: exact signal reconstruction from
  highly incomplete frequency information,''
\newblock {\em {{IEEE} {T}rans. {I}nf. {T}heory}}, vol. 52, no. 2, pp.
  489--509, 2006.

\bibitem{Donoho06}
D.~L. Donoho,
\newblock ``{Compressed sensing},''
\newblock {\em {{IEEE} {T}rans. {I}nf. {T}heory}}, vol. 52, no. 4, pp.
  1289--1306, Apr. 2006.

\bibitem{Rauhut10}
H.~Rauhut,
\newblock ``{C}ompressive {S}ensing and {S}tructured {R}andom {M}atrices,''
\newblock in {\em {T}heoretical {F}oundations and {N}umerical {M}ethods for
  {S}parse {R}ecovery}, {M}. {F}ornasier, Ed., vol.~9 of {\em {R}adon {S}eries
  {C}omp. {A}ppl. {M}ath.}, pp. 1--92. de{G}ruyter, 2010.

\bibitem{Candes11}
E.~J. Cand{\`e}s and Y.~Plan,
\newblock ``A probabilistic and ripless theory of compressed sensing,''
\newblock {\em {{IEEE} {T}rans. {I}nf. {T}heory}}, vol. 57, no. 11, pp.
  7235--7254, 2011.

\bibitem{Puy11}
G.~Puy, P.~Vandergheynst, and Y.~Wiaux,
\newblock ``On variable density compressive sampling,''
\newblock {\em {{IEEE} {S}ignal {P}rocessing {L}etters}}, vol. 18, no. 10, pp.
  595--598, 2011.

\bibitem{Juditsky11}
A.~Juditsky and A.~Nemirovski,
\newblock ``On verifiable sufficient conditions for sparse signal recovery via
  $\ell_1$ minimization,''
\newblock {\em Mathematical Programming Ser. B}, vol. 127, pp. 89--122, 2011.

\bibitem{Juditsky11b}
A.~Juditsky, F.K. Karzan, and A.~Nemirovski,
\newblock ``On low rank matrix approximations with applications to synthesis
  problem in compressed sensing,''
\newblock {\em SIAM J. on Matrix Analysis and Applications}, vol. 32, no. 3,
  pp. 1019--1029, 2011.

\bibitem{ledoux01}
M.~Ledoux,
\newblock ``The {C}oncentration of {M}easure {P}henomenon,''
\newblock {\em Amer. Mathematical Society}, vol. 89, 2001.

\bibitem{hastings1970montecarlo}
W.~K. Hastings,
\newblock ``{Monte Carlo sampling methods using Markov chains and their
  applications},''
\newblock {\em Biometrika}, vol. 57, no. 1, pp. 97--109, Apr. 1970.

\bibitem{Lezaud98}
P.~Lezaud,
\newblock ``Chernoff-type bound for finite {M}arkov chains,''
\newblock {\em Annals of Applied Probability}, vol. 8, no. 3, pp. 849--867,
  1998.

\bibitem{Horn91}
R.~Horn and C.~Johnson,
\newblock {\em {Topics in matrix analysis}},
\newblock Cambridge University Press, Cambridge, 1991.

\bibitem{Combettes11b}
P~L Combettes and J.-C Pesquet,
\newblock ``{Proximal Splitting Methods in Signal Processing},''
\newblock in {\em {Fixed-Point Algorithms for Inverse Problems in Science and
  Engineering}}, pp. 185--212. Springer, 2011.

\end{thebibliography}

}

\end{document}